\documentclass[journal,10 pt]{IEEEtran} 

\usepackage{graphicx}
\usepackage{textcomp}
\usepackage{epsfig} 
\usepackage{epsfig,color,amsmath,cite}
\usepackage{amsthm} 
\usepackage{amsmath}    
\usepackage{bm}
\usepackage{epstopdf}
\usepackage{amssymb}
\usepackage{url}
\usepackage{multirow}
\usepackage{hhline}
\usepackage{booktabs}
\usepackage[linesnumbered,boxed,commentsnumbered,ruled,vlined,longend]{algorithm2e}
\usepackage{comment}

\makeatother
\DeclareMathAlphabet\mathbfcal{OMS}{cmsy}{b}{n}

\newtheorem{mydef}{Definition}

\newtheorem{myrem}{Remark}
\newtheorem{asmp}{Assumption}

\newtheorem{myprs}{Proposition}



\usepackage{stackengine}

\newcommand{\mat}[1]{\boldsymbol{#1}}

\providecommand{\eye}{\mat{I}}
\providecommand{\mA}{\ensuremath{\mat{A}}}

\providecommand{\mG}{\ensuremath{\mat{G}}}

\newcommand{\m}{\boldsymbol}
\allowdisplaybreaks[4]
\usepackage[colorlinks = true,
linkcolor = blue,
urlcolor  = blue,
citecolor = blue,
anchorcolor = blue]{hyperref}

\newcommand{\mr}[1]{\mathrm{#1}}
\usepackage[framemethod=TikZ]{mdframed}
\mdfdefinestyle{MyFrame}{%
	linecolor=black,
	outerlinewidth=1.25pt,
	roundcorner=1.25pt,
	innerrightmargin=5pt,
	innerleftmargin=5pt,}
	
\usepackage[noabbrev]{cleveref}

\usepackage{mathtools}

\DeclarePairedDelimiter\abs{\lvert}{\rvert}%
\DeclarePairedDelimiter\norm{\lVert}{\rVert}%

\makeatletter
\let\oldabs\abs
\def\abs{\@ifstar{\oldabs}{\oldabs*}}
\let\oldnorm\norm
\def\norm{\@ifstar{\oldnorm}{\oldnorm*}}
\makeatother

\usepackage[english]{babel}
\usepackage[utf8]{inputenc}
\usepackage[super]{nth}

\usepackage{graphicx}
\usepackage{float}
\usepackage[caption = false]{subfig}

\usepackage{array}
\usepackage{threeparttable}

\usepackage[english]{babel}
\usepackage[utf8]{inputenc}
\usepackage[super]{nth}

\usepackage{pifont}

\usepackage{mathtools}

\usepackage{mleftright}
\usepackage{xparse}

\NewDocumentCommand{\evaluat}{sO{\big}mm}{%
	\IfBooleanTF{#1}
	{\mleft. #3 \mright|_{#4}}
	{#3#2|_{#4}}%
}

\usepackage{lettrine} 

\usepackage{balance} 

\title{\Large \vspace{1cm} \LARGE \centering {\textsc{\textbf{ODE Transformations of Nonlinear DAE Power Systems*}}}}

	\author{Mohamad H. Kazma, \textit{Graduate Student Member, IEEE} and Ahmad F. Tah$\text{a}^{\diamond}$, \textit{Member, IEEE} 
\vspace{-0.2cm}
\thanks{
	$^\diamond$Corresponding author.  This work is supported by National Science Foundation under Grants 2152450 and 2151571. The authors are with  the Civil \& Environmental Engineering and Electrical \& Computer Engineering Departments, Vanderbilt University, 2201 West End Ave, Nashville, Tennessee 37235. Emails: mohamad.h.kazma@vanderbilt.edu, ahmad.taha@vanderbilt.edu.\newline
	*The discrepancy between this version of the paper and the PES General Meeting 2024 version (available on IEEE Explore) is due to the imposed page limit for all conference papers. This version includes additional comments as well as an easier paper layout for better readability.}
	}
	
\begin{document}



\maketitle

\markboth{To appear in the IEEE PES General Meeting 2024, Seattle, Washington, July 2024}{}

\begin{abstract}
	Dynamic power system models are instrumental in real-time stability, monitoring, and control. Such models are traditionally posed as systems of nonlinear differential algebraic equations (DAEs): the dynamical part models generator transients and the algebraic one captures network power flow. While the literature on control and monitoring for ordinary differential equation (ODE) models of power systems is indeed rich, that on DAE systems is \textit{not}. DAE system theory is less understood in the context of power system dynamics. To that end, this paper presents two new mathematical transformations for nonlinear DAE models that yield nonlinear ODE models while retaining the complete nonlinear DAE structure and algebraic variables. Such transformations make (more accurate) power system DAE models more amenable to a host of control and state estimation algorithms designed for ODE dynamical systems. We showcase that the proposed models are effective, simple, and computationally scalable. 
\end{abstract}

\section{Introduction}\label{sec:Introduction}
\lettrine[lines=2]{P}{ower} systems monitoring, state estimation, control, and transient stability analysis are all reliant on high-fidelity models of multi-machine power systems.
In power grids, transient stability analysis determines how the power system maintains synchronicity under time-varying conditions and large uncertainties from load disturbances~\cite{Aolaritei2018,Cui2022}.
Such analysis relies on time-domain simulations of the system that is expressed as a set of differential algebraic equations (DAEs)~\cite{Milano2016,Milano2016a}.
The differential-algebraic nature couples the system dynamics with power flow constraints, thus resulting in a more accurate model. Nonlinear DAEs are an extreme case of stiff dynamical systems~\cite{Guanabaraa}, meaning that the system has time constants that span several orders of magnitude; in particular, the algebraic constraints exhibit null time constants.

In general, nonlinear DAEs are solved using implicit discretization schemes~\cite{Astic1994b}. 
Multi-step methods offer stable and efficient schemes when dealing with nonlinear DAEs~\cite{Brayton1972}. 
Such discrete-time modeling methods include: backward differential formulas (BDF)~\cite{Gear1971, Milano2022}, backward Euler (BE) method~\cite{Milano2022}, and trapezoidal implicit (TI) method~\cite{Potra2006,Milano2022}. Simulating discrete-time models requires an integrative time-step algorithm~\cite{Lara2023}, and the solvability of power system DAEs under implicit integration methods is well-established. The Newton-Raphson (NR) method~\cite{Conejo2018, Milano2022} is generally implemented within power system simulation packages to solve discretized DAEs~\cite{Sauer2017}. Despite such well-developed time-domain numerical solutions, from a systems' theory perspective, the literature on nonlinear DAE power networks is limited---unlike that of ordinary differential equations (ODE models)~\cite{Grob2016}.


Common modeling for systems' control and estimation is based on an ODE formulation; this is due to the aforementioned limitation on system theory for DAEs. Typically, the formulation of ODE systems from DAE models is performed by either neglecting the algebraic constraints or by formulating a decoupled modeling approach~\cite{Nugroho2022}. The simplified models potentially limit the transient stability simulations and, ultimately, the estimation and control performance. Time-domain simulations resulting from the full DAE models of a power network can, for instance, give an accurate depiction of the dynamics under topological changes triggered by faults and be modeled to include uncertain loads from renewable energy resources. 

The limitation on model fidelity from a control and estimation perspective is expressed in the form of the following research question: \textit{How do we extend existing systems control theory, developed for ODE dynamical systems, to accurately apply it to DAE models?}
Descriptor systems---arising from DAEs models---appear in numerous applications, with a few examples being chemical, electrical and mechanical systems. For such reason, there is a rise in interest towards translating control and stability theory to the analysis of descriptor systems~\cite{Jaiswal2022}. 

Recent studies---see,~\cite{Kunkel2008,Althoff2014,Liu2020,Jaiswal2022,Nugroho2022, Nadeem2022} 
and reference therein---present literature on developing the state estimation and control theory of DAE systems, in particular that of linear DAEs. 
However, in this paper, we aim to address the limitations on state estimation, control, and transient stability analysis of power systems by giving a new perspective on DAE to ODE system modeling. Therefore, we instead attempt to address the posed research question in the form of: \textit{Is there a methodology to accurately restructure the DAE system into an ODE model without loss of information and therefore exploit existing ODE systems control theoretic?}

To that end, in this paper we introduce two simple yet effective methods to transform nonlinear DAE power system models into ODEs. The idea is to formulate ODE-structured representations of the network dynamics while depicting the full nonlinear DAE structure along with the algebraic constraints. These transformations allow the utilization of the rich literature on control and estimation of ODE models of power systems.
 The main contributions of this paper are:
\begin{itemize}
\item We present two mathematical transformations that retain the complete nonlinear DAE dynamics while achieving a nonlinear ODE-structure. The first transformation is based on applying the implicit function theorem (IFT) to reformulate the algebraic constraints into an ODE model (Section~\ref{sec:ODE-model}). The second transformation is approximation-based; it yields an effective approximation towards ODE-structured algebraic constraints (Section~\ref{sec:app-ODE-model}).
\item We show that it is viable to model and perform both continuous and discrete time-domain simulations for transient stability analysis on the proposed systems. We also illustrate that the resulting Hessian matrices computations that arise when simulating the discretized IFT transformed system can be approximated using finite difference approximations (Section~\ref{sec:sim-power-system}).
\item We assess the validity of the proposed models on standard power networks under transient time-domain simulations subject to load disturbances (Section~\ref{sec:case-studies}).\vspace{-0.1cm}
\end{itemize}


We note that, in this paper, we do not focus on extending DAE systems stability and control theory, or on formulating new discretization techniques for such models, but rather on introducing the aforementioned transformations onto the nonlinear DAE power system. The transformations' significance is accentuated by the fact that we obtain nonlinear ODE structured models that retain the complete nonlinear DAE structure along with the algebraic variables. Utilizing the ODE structured models, the rich literature on control and monitoring of ODE dynamical systems can be exploited.

\section{Proposed Nonlinear DAE Transformations}\label{sec:nonlinearmodel} 
In this paper, we consider a general nonlinear model of multi-machine power system dynamics $(\mathcal{N} ,\mathcal{E})$, where $\mathcal{E} \subseteq \mathcal{N} \times \mathcal{N}$ is the set of transmission lines, $\mathcal{N} = \mathcal{G} \cup \mathcal{L}$ is the set of all buses in the network, while $\mathcal{G}$ and $\mathcal{L}$ are the set of generator and load buses, respectively.
The model represents both the generator dynamics and algebraic constraints. Readers can refer to~\cite[Ch. 7]{Sauer2017} for the full description of the power network utilized within this work. Note that the theoretical developments herein still apply to any semi-explicit nonlinear DAE model of a power system written as
\begin{subequations}~\label{eq:semi_NDAE_rep}
\begin{align}
	\textit{generator dynamics}:	\;\; \dot{\m x}_{d} &=  \m{f}(\m x_d,\m x_a, \m u), \label{X_d} \\
	\textit{algebraic constraints}:	\;\; \m{0} & = \m{g}(\m x_d, \m x_a),\label{X_a}
\end{align} 
\end{subequations}
where the dynamic states of the generator are defined as $\m{x}_{d} := \m{x}_{d}(t) \in \mathbb{R}^{n_d}$, the algebraic states as $\m{x}_{a} := \m{x}_{a}(t) \in \mathbb{R}^{n_a}$ and the system input as $\m{u}:=\m{u}(t) \in \mathbb{R}^{n_u}$. 
Functions $\m{f}(\cdot): \mathbb{R}^{n_d}\times\mathbb{R}^{n_a}\times\mathbb{R}^{n_u}\rightarrow\mathbb{R}^{n_d}$ and $\m{g}(\cdot): \mathbb{R}^{n_d}\times\mathbb{R}^{n_a}\rightarrow\mathbb{R}^{n_a}$ are nonlinear and define the system dynamics and power flows. 

The existence of a solution for nonlinear DAEs can be determined by proving that a DAE is \textit{strangeness-free}~\cite[Hypothesis 4.2]{Volker2005}, i.e., the strangeness index is equal to zero. This index is a generalization of the differentiation index of DAEs. Refer to~\cite{Volker2005} for the detailed hypothesis that defines the strangeness index of nonlinear DAEs.
\begin{mydef}~\label{def:inde(1)}
The differentiation index~\cite{Campbell1995,Volker2005,Chen2021} 
of descriptor system refers to the number of differentiations required to obtain ODEs using algebraic manipulations.
\end{mydef}

For linear DAE systems, a set of differential algebraic equations is of index one if and only if it is regular. Regularity is an important property for linearized DAEs; it is a condition for the existence of a consistent unique solution for every initial condition~\cite{Grob2016}.
\begin{mydef}~\label{def:regular}
Regularity of a linearized DAE~\eqref{eq:semi_NDAE_rep} around an initial state can be characterized by matrix pair $(\m{E}, \m{A})$, such that it is regular if and only if $\mr{det}(s\m{E} - \m{A}) \neq 0 $ for $s \in \mathbb{C}$.
\end{mydef}

The linearized representation around an operating point of the power system DAE~\eqref{eq:semi_NDAE_rep} can be rewritten as
\begin{equation}~\label{eq:semi_NDAE_rep_linear}
\m{E}\dot{\m{x}} = \m{A}\m{x} + \m{B}\m{u},
\end{equation}
where $\m{E} \in \mathbb{R}^{n_d+n_a}$ represents the singular mass matrix that has ones on its diagonal entry for the differential equations and zeros for the algebraic equations. The constant state-space matrices are defined as $\m{A} \in \mathbb{R}^{n_d+n_a}$ and $\m{B}\in \mathbb{R}^{n_u}$.

In~\cite{Nugroho2022}, it is shown that the linearized representation of power system~\eqref{eq:semi_NDAE_rep} has a differentiation index of one and is regular. Such condition guarantees that for each consistent initial condition, a unique solution exists. Under the linearized model, this condition holds true if and only if there exist paths whereby every load is connected to a generator bus.
To prove a solution is unique for nonlinear DAEs is rather complex and still considered an open problem~\cite{Grob2016}. 
We note here that proving the linearized dynamic system~\eqref{eq:semi_NDAE_rep_linear} is regular and of index one is a prerequisite for providing evidence that the nonlinear system can be strangeness-free. The differentiation index, regularity and strangeness-free property of a nonlinear DAE system ensure the solvability of the system.
Considering the aforementioned results from ~\cite{Nugroho2022} and the complexity of proving such condition for nonlinear DAEs, the following assumption holds true within this paper. 
\begin{asmp}~\label{assump:index}
The DAE~\eqref{eq:semi_NDAE_rep} is strangeness-free, of differentiation index one and regular. Thus, under any consistent initial conditions a unique solution exists and the partial derivatives of $\m{f}(\cdot)$ and $\m{g}(\cdot)$ with respect to $\m{x}_{d}$ and $\m{x}_{a}$ are non-singular~\cite{Volker2005}. 
\end{asmp}

The aforementioned assumption is mild and holds true for the test cases considered in the numerical studies section. Under such conditions, the two transformations that reformulate~\eqref{eq:semi_NDAE_rep} into nonlinear ODEs can be posed. 
\begin{myrem}\label{rem:Solution}
The time-domain numerical solvability of the DAE system~\eqref{eq:semi_NDAE_rep} implies that a unique solution for different inputs and consistent initial conditions exists~\cite{Crow2015}.
\end{myrem}

Hence, the time-domain simulations presented in Section~\ref{sec:case-studies} validate Assumption~\ref{assump:index} for the power system represented as~\eqref{eq:semi_NDAE_rep} and therefore provide explicit proof that the nonlinear power system DAE is strangeness-free, of differentiation index one and regular.


\subsection{IFT-Based Nonlinear DAE Model: $\mr{ODE}$-$\mr{DAE}$}\label{sec:ODE-model} 
The first transformation relies on applying the IFT~\cite[Theorem 3.3.1]{Krantz2013} to resolve the algebraic constraints into ODEs $(\mr{ODE}$-$\mr{DAE})$. This method entails differentiating the algebraic constraints in~\eqref{X_a} with respect to time variable $t$. 	
\begin{myprs}~\label{propo:ODE-DAE}
Consider the nonlinear descriptor system~\eqref{eq:semi_NDAE_rep} and that Assumption~\ref{assump:index} holds true for any initial condition and control input, then the descriptor system can be restructured into a nonlinear ODE system by applying the IFT~\cite[Theorem 3.3.1]{Krantz2013} to resolve the algebraic constraints into ODEs. As such, the nonlinear descriptor system~\eqref{eq:semi_NDAE_rep} can be rewritten as
\begin{subequations}~\label{eq:semi_NDAE_rep-ODE}
	\begin{align}
		\dot{\m x}_{d} &=  \m{f}(\m x_d,\m x_a, \m u), \label{X_d_ode} \\
		\dot{\m{x}}_a & = 	-(\mG_{\m{x}_a})^{-1}\mG_{\m{x}_d}\m{f}(\m{x}_d,\m{x}_a,\m{u}) = \tilde{\m{g}}(\m{x}_d,\m{x}_a,\m{u}).\label{X_a_ode}	\vspace{-0.1cm}
	\end{align} 
\end{subequations}
\end{myprs}	
\begin{proof}
Under Assumption~\ref{assump:index}, the partial derivative $\frac{\partial \m{g}(\m{x}_d,\m{x}_a) }{\partial \m{x}_a}$ is non-singular, refer to~\cite{Grob2016,Nugroho2022}. Then, implicitly differentiating the algebraic constraints~\eqref{X_a} using the implicit function theorem we obtain the following
\begin{equation}~\label{eq:partial_alg_const}
	0 = \frac{\partial \m{g}(\m{x}_d,\m{x}_a) }{\partial \m{x}_d}\frac{\partial \m{x}_d }{\partial t} +\frac{\partial \m{g}(\m{x}_d,\m{x}_a) }{\partial \m{x}_a}\frac{\partial \m{x}_a }{\partial t}.
\end{equation}

We define the Jacobian matrices of the implicit algebraic constraints with respect to $\m{x}_a$ and $\m{x}_d$ as $\mG_{\m{x}_a} :=\frac{\partial \m g(\m{x}_d,\m{x}_a) }{\partial \m{x}_a}$ $\in \mathbb{R}^{n_{d}\times n_{a}}$  and $\mG_{\m{x}_d} :=\frac{\partial \m g(\m{x}_d,\m{x}_a) }{\partial \m{x}_d}$ $\in \mathbb{R}^{n_{a} \times n_{d}}$. The time-derivative $\frac{\partial \m{x}_d }{\partial t}:= \dot{\m{x}}_d = \m{f}(\m x_d,\m x_a, \m u)$, refer to~\eqref{X_d}. As such, the algebraic constraints~\eqref{X_a} can be rewritten as~\eqref{X_a_ode}.
	\end{proof}
	

The resulting $\mr{ODE}$-$\mr{DAE}$~\eqref{X_a_ode} representing the algebraic constraints~\eqref{X_a} is now a function of~\eqref{X_d} and therefore depends on control input $\m{u}$. Nevertheless~\eqref{X_a} can be explicitly formulated and then differentiated with respect to $t$, however this requires the complex task of explicitly reconstructing the algebraic equations---which is rather difficult to perform. In contrast, the IFT method allows us to utilize the semi-implicit system~\eqref{eq:semi_NDAE_rep} along with all the dynamic and algebraic relationships that are inherent to the system while also being represented as a set of ODEs. 

It is important to note that this index reduction approach can be applied to nonlinear DAEs of differentiation index $n$. In other words, this theoretical method can also be applied for higher index DAEs, however it requires several rounds of differentiation to reformulate the algebraic constraints into a set of ODE equations.
The application on DAE of index greater than one is outside the scope of this paper and power systems in general---power systems are typically of index one.
\subsection{Approximate Nonlinear DAE Model: $\mr{Approx}$-$\mr{DAE}$}\label{sec:app-ODE-model} 
The second transformation formulates an approximate DAE model $(\mr{Approx}$-$\mr{DAE})$ that is based on introducing a positive scalar term, denoted by $\mu$, to the DAE system at the algebraic constraint equations level. In particular,  the left-hand side of~\eqref{X_a} is replaced by $\mu \dot{\m{x}}_{a}$.  The rationale behind utilizing such simplistic alternative as compared to the first transformation is evident at the discrete-time modeling level; it offers an alternative to dealing with the Hessian matrix computations that arise under the $\mr{ODE}$-$\mr{DAE}$ approach---refer to section~\ref{sec:sim-power-system}. Under such transformation, $\mu > 0$ is defined as a relatively small term that simulates the system's dynamics while satisfying the power flow constraints. As a result, the descriptor system~\eqref{eq:semi_NDAE_rep} can be rewritten as
\begin{subequations}~\label{eq:mu-semi_NDAE_rep}
	\begin{align}
		\dot{\m x}_{d} &=  \m{f}(\m x_d,\m x_a, \m u), \label{X_d_mu} \\
		\mu\dot{\m{x}}_a  & = \m{g}(\m x_d, \m x_a) +\m{O}(\mu),\label{X_a_mu}
	\end{align} 
\end{subequations}
where the approximation error $\m{O}(\mu)$ is of order $\mu$, such that as $\mu \rightarrow 0$ the error becomes null.
\begin{myrem}~\label{rem:mu}
	The proposed $\mr{Approx}$-$\mr{DAE}$ model is not the model reduction technique commonly referred to as the singular perturbation technique
	~\cite{Chow1990,McClamroch1994,Shen2020,Lara2023}.
	Singular perturbation is applied to dynamical systems that have multiple time-scale dynamics, including electrical systems, to reduce the number of simulated states; it results in a simplified model representation. 
\end{myrem}
\begin{table}[t]
	\centering 
	\caption{Discretized Dynamics of the Algebraic Constraints}
	\label{tab:DAE-ODE_disc}
	\begin{tabular}{l|l}
		\midrule \hline
		$\hspace{-0.15cm}\mr{DAE}\hspace{-0.15cm}$
		&\multicolumn{1}{c}{$\hspace{-0.15cm} \m{0} \hspace{-0.05cm} = \hspace{-0.05cm} \begin{cases}
				\hspace{-0.05cm} - \tilde{h} \m{g}(\m{x}_{k})
				& \hspace{-0.15cm} \text{for BDF}\\
				\hspace{-0.05cm}  - \tilde{h} (\m{g}(\m{x}_{k})+\m{g}(\m{x}_{k-1}))
				& \hspace{-0.15cm} \text{for TI}
			\end{cases} \hspace{-0.1cm} $} \\ \hline 
		$\hspace{-0.15cm}\mr{ODE}$-$\mr{DAE}\hspace{-0.15cm}$
		&\multicolumn{1}{c}{$\hspace{-0.15cm} \m{0} \hspace{-0.05cm} = \hspace{-0.05cm} \begin{cases}
				\hspace{-0.05cm} \m{x}_{a,k} -\Sigma_{s=1}^{k_g}\alpha_s \m {x}_{a,k-s} - \tilde{h} \tilde{\m{g}}(\m{z}_{k}) 
				&\hspace{-0.15cm} \text{for BDF}\\
				\hspace{-0.05cm} \m{x}_{a,k} -\m {x}_{a,k-1} - \tilde{h} (\tilde{\m{g}}(\m{z}_{k})+\tilde{\m{g}}(\m{z}_{k-1}))
				&\hspace{-0.15cm} \text{for TI}
			\end{cases} \hspace{-0.1cm} $} \\ \hline
		$\hspace{-0.15cm}\mr{Approx}$-$\mr{DAE} \hspace{-0.15cm}$
		&\multicolumn{1}{c}{$\hspace{-0.15cm} \m{0} \hspace{-0.05cm} = \hspace{-0.05cm} \begin{cases}
				\hspace{-0.05cm} \mu \m{x}_{a,k} -\Sigma_{s=1}^{k_g}\alpha_s \mu \m {x}_{a,k-s} - \tilde{h} \m{g}(\m{x}_{k})
				\hspace{-0.1cm}& \hspace{-0.15cm} \text{for BDF}\\
				\hspace{-0.05cm} \mu\m{x}_{a,k} - \mu\m{x}_{a,k-1}- \tilde{h} (\m{g}(\m{x}_{k})+\m{g}(\m{x}_{k-1})) \hspace{-0.1cm}& \hspace{-0.15cm} \text{for TI}
			\end{cases} \hspace{-0.25cm} $} \\
		\toprule \bottomrule
	\end{tabular}
\end{table}

Herein, we utilize the small positive scalar $\mu$ to reduce to stiffness of the DAE model, i.e, transforming the null time-scale constants for the existing algebraic constraints to become of order $\mu$.
The introduced dynamic $\mu \dot{\m{x}}_a$ is an arbitrary modification to the system's behavior and depends highly on the choice of $\mu$. The value of $\mu$ contributes to the stiffness of the system, i.e., the order of $\mu$ defines the time-scale for the algebraic equations that are modeled as dynamic equations. Such that as $\mu \rightarrow 0$, the time-scale becomes null, and therefore the system reverts to having algebraic equations that are extremely stiff. 

\section{Time-domain Modeling for Transient Stability Analysis}~\label{sec:sim-power-system}
The choice of discretization method must rely upon 
the system's stiffness, desired accuracy, and the performance of computation resources. 
Stiff dynamical systems can be identified from time constants on local subsystems that have contrasting magnitudes by a large margin. 
Nonlinear dynamic power systems under transient conditions are in practice modeled as discrete-time state-space models and are solved using numerical methods~\cite{Crow2015}.
Nonlinear DAEs exhibit stiff dynamics and are solved using implicit discretization methods that offer stable computational methods as compared with explicit methods~\cite{Astic1994b}. 

In this work we approach discretizing the nonlinear DAE system and the proposed transformations using BDF, BE, and TI discrete-time modeling methods. Specifically, we investigate how the two transformations perform when embedded within vintage discretization methods. 
BDF depends on its discretization index denoted by $k_g$, which is stable for $ 2\leq k_g \leq 5$. We note that for $k_g=1$, BDF renders into the BE method. For succinctness, we refer to BE as BDF while designating the order of $k_g = 1$. We start by discretizing the differential dynamics~\eqref{X_d} that are common between the presented models. 
To that end, the discrete-time representation of the generator dynamics represented in~\eqref{X_d} can be written as	
\begin{equation}~\label{eq:disc_ssm_NDAE}
	\hspace{-0.3cm} \m{0} =
	\begin{cases}
		\hspace{-0.05cm} \m{x}_{d,k} -\Sigma_{s=1}^{k_g}\alpha_s \m {x}_{d,k-s} - \tilde{h} \m{f}(\m{z}_{k})& \hspace{-0.1cm} \text{for BDF,}\\
		\hspace{-0.05cm} \m{x}_{d,k} -\m {x}_{d,k-1} - \tilde{h} (\m{f}(\m{z}_{k})+\m{f}(\m{z}_{k-1}))&\hspace{-0.1cm} \text{for TI,}
	\end{cases} \hspace{-0.4cm} \\
\end{equation}
where vector $\m{z}_{k} :=[\m{x}_{d,k}, \m{x}_{a,k},\m{u}_{k}]^{\top}$ and $\m{x}_{k} :=[\m{x}_{d,k}, \m{x}_{a,k}]^{\top}$ for time step $k$. The discretization time step size $\tilde{h}$ is defined as $\tilde{h}  := \beta h \;\; \text{for} \;\text{BDF and}\; 0.5h  \;\; \text{for} \; \text{TI}$, where $h$ is the simulation time step size.


The discretization constants $\beta$ and $\alpha_s$ for BDF method depend on the order of index $k_g$ and are calculated as 
	\begin{equation}~\label{eq:beta_alpha}
			\beta = \Big(\sum_{s=1}^{k_g} \frac{1}{s}\Big)^{-1}  \;,\;\;\; \alpha_s = (-1)^{(s-1)} \beta \sum_{j=s}^{k_g} \frac{1}{j}  \begin{pmatrix}
					j \\ s
				\end{pmatrix}.
		\end{equation}

The implicit discrete-time representation of the algebraic constraints for the descriptor system~\eqref{eq:semi_NDAE_rep} and the proposed models are summarized in Table~\ref{tab:DAE-ODE_disc}.
Solving discrete-time models requires an integrative time-step algorithm~\cite{Lara2023}. The Newton-Raphson (NR) method~\cite{Conejo2018, Milano2022} is generally implemented to solve the implicit discrete-time nonlinear descriptor dynamics.

\begin{table}[t]
	\centering 
	\caption{Variations in the Jacobian of ODE-Transformed Systems}
	\label{tab:Jac-DAE-ODE_disc}
	\renewcommand{\arraystretch}{2}
	\begin{tabular}{l|l}
		\hline \hline 
		$\hspace{-0.15cm}\mr{ODE}$-$\mr{DAE}\hspace{-0.15cm}$
		&\multicolumn{1}{c}{$\begin{bmatrix}
				-\tilde{h}\widetilde{\m{G}}_{\m{x}_{d}}(\m{z}^{(i)}_{k}) & \eye_{n_a}-\tilde{h} \widetilde{\m{G}}_{\m{x}_{a}}(\m{z}^{(i)}_{k})
			\end{bmatrix}$} \\ \hline
		$\hspace{-0.15cm}\mr{Approx}$-$\mr{DAE} \hspace{-0.15cm}$
		&\multicolumn{1}{c}{$\begin{bmatrix}
				-\tilde{h}\m{G}_{\m{x}_{d}}(\m{x}^{(i)}_{k}) & \mu\eye_{n_a}-\tilde{h} \m{G}_{\m{x}_{a}}(\m{x}^{(i)}_{k})
			\end{bmatrix}$} \\
		\toprule \bottomrule
	\end{tabular}
\end{table}

To implement the NR method, the Jacobian pertaining to the nonlinear dynamics is evaluated. At each time step $k$, the increment $\Delta \m{x}_{k}^{(i)}$ defined as~\eqref{eq:mu-newton_raph} is evaluated and used to update the state variables $\m{x}^{(i+1)}_{k}= \m{x}_{k}^{(i)} + \Delta \m{x}_{k}^{(i)}$ under the NR iteration $i$ until a convergence criterion is satisfied. Once NR iteration converges, time step $k$ advances until the dynamics over time span $t$ is simulated. The iteration increment $\Delta \m{x}^{(i)}_{k}$ can be written as
\begin{equation}~\label{eq:mu-newton_raph}
	\Delta \m{x}^{(i)}_{k} = \left[\mA_{g}(\m{z}^{(i)}_{k})\right]^{-1}\begin{bmatrix}
		\m{\phi}(\m{z}^{(i)}_{k})
	\end{bmatrix},
\end{equation}
where $\m{z}^{(i)}_{k} :=[\m{x}^{(i)}_{d,k},\m{x}^{(i)}_{a,k},\m{u}^{(i)}_{k}]$, and $\m{\phi}(\m{z}^{(i)}_{k})$ denotes the discretized system dynamics, represented in~\eqref{eq:disc_ssm_NDAE} and Table~\ref{tab:DAE-ODE_disc}, under the NR iteration index $i$. The Jacobian of $\m{\phi}(\m{z}^{(i)}_{k})$ is defined as $\m{A}_{g}(\m{z}^{(i)}_{k}) := \begin{bmatrix}
	\tfrac{\partial \m{\phi}(\m{z}^{(i)}_{k})}{\partial \m{x}_k}
\end{bmatrix}$ and can be written as
\begin{equation}~\label{eq:Jac_Newton_Raph}
	\hspace{-0.3cm} \m{A}_{g}(\m{z}^{(i)}_{k})
	\hspace{-0.05cm} = \hspace{-0.05cm} \begin{bmatrix}
		\eye_{n_d}-\tilde{h} \m{F}_{\m{x}_{d}}(\m{z}^{(i)}_{k}) & -\tilde{h} \m{F}_{\m{x}_{a}}(\m{z}^{(i)}_{k}) \\
		-\tilde{h}\m{G}_{\m{x}_{d}}(\m{x}^{(i)}_{k}) & -\tilde{h} \m{G}_{\m{x}_{a}}(\m{x}^{(i)}_{k})
	\end{bmatrix},
\end{equation}
where matrices $\m{F}_{\m{x}_{d}}(\cdot) \in \mathbb{R}^{n_{d}\times n_{d}}$ and $\m{F}_{\m{x}_{a}}(\cdot)$ $\in \mathbb{R}^{n_{d}\times n_{a}}$ are the Jacobians of~\eqref{eq:disc_ssm_NDAE} with respect to $\m{x}_{d}$ and $\m{x}_{a}$. Matrix $\eye_{n_d}$ is an identity matrix of dimension similar to $\m{F}_{\m{x}_{d}}(\cdot)$. Matrices $\m{G}_{\m{x}_{d}}(\cdot)$ and $\m{G}_{\m{x}_{a}}(\cdot)$ retain the same definition as before, however under the discretized dynamics presented in Table~\ref{tab:DAE-ODE_disc}.
\begin{figure}[h]
	\centering
	\hspace{-0.2cm}
	\subfloat{\includegraphics[keepaspectratio=true,scale=0.4]{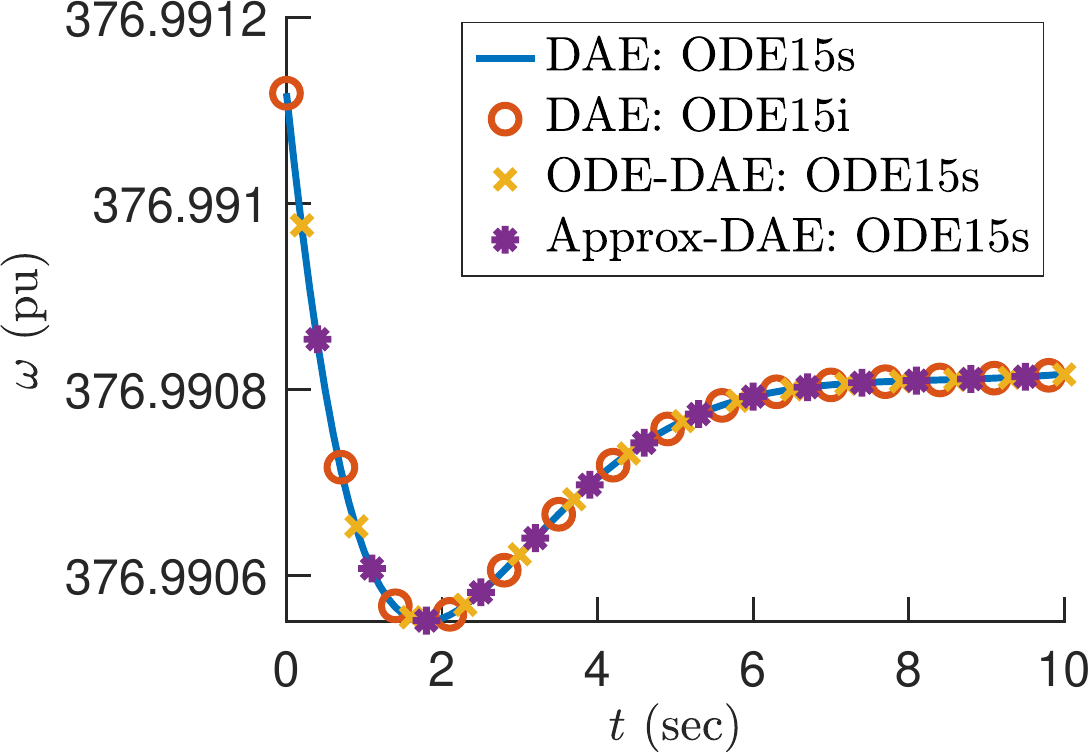}}\vspace*{-0.35cm} 	\hspace{-0.2cm}
	\subfloat{\includegraphics[keepaspectratio=true,scale=0.4]{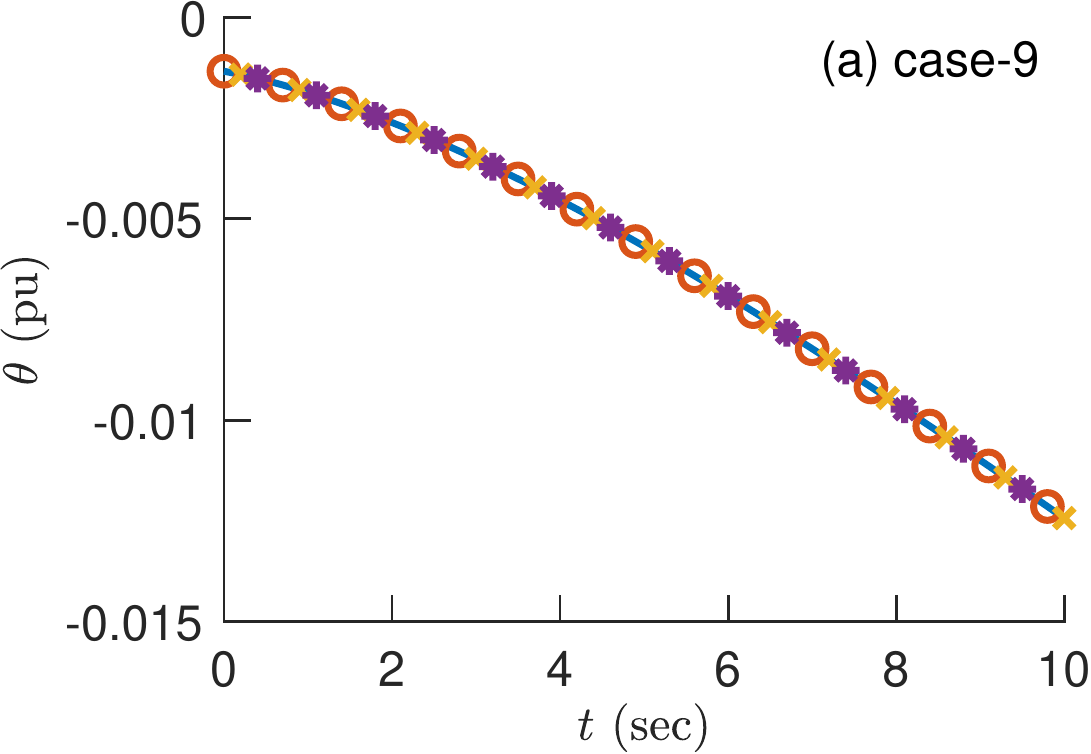}}{}{}
	\subfloat{\includegraphics[keepaspectratio=true,scale=0.4]{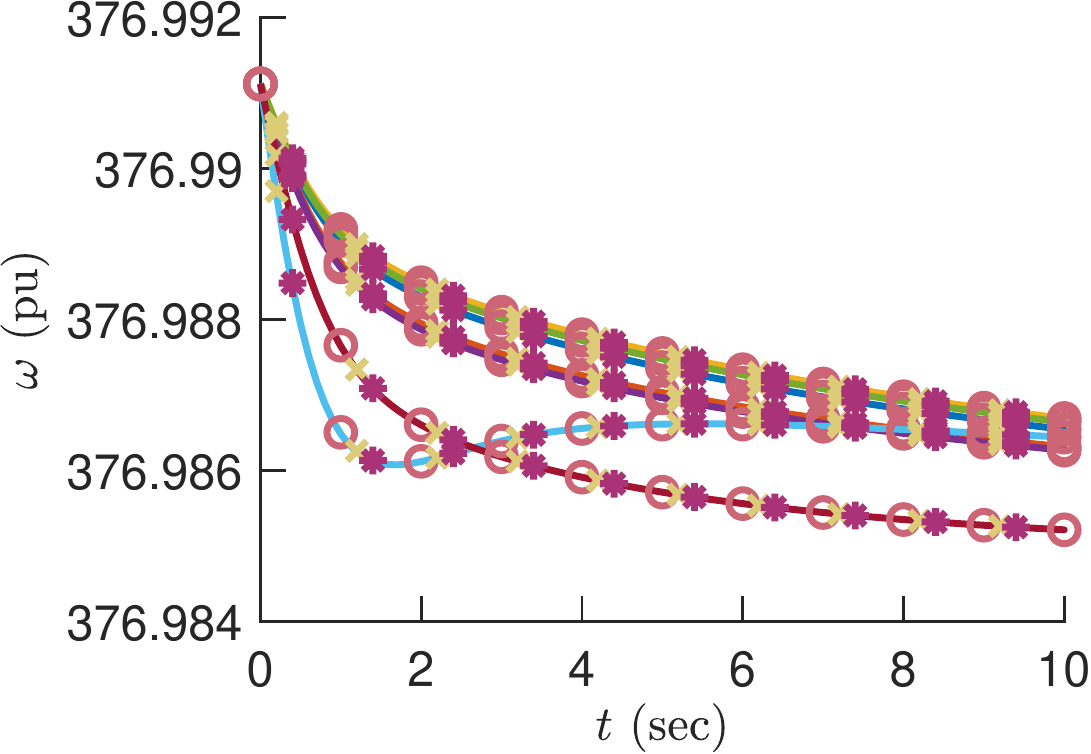}}
	\subfloat{\includegraphics[keepaspectratio=true,scale=0.4]{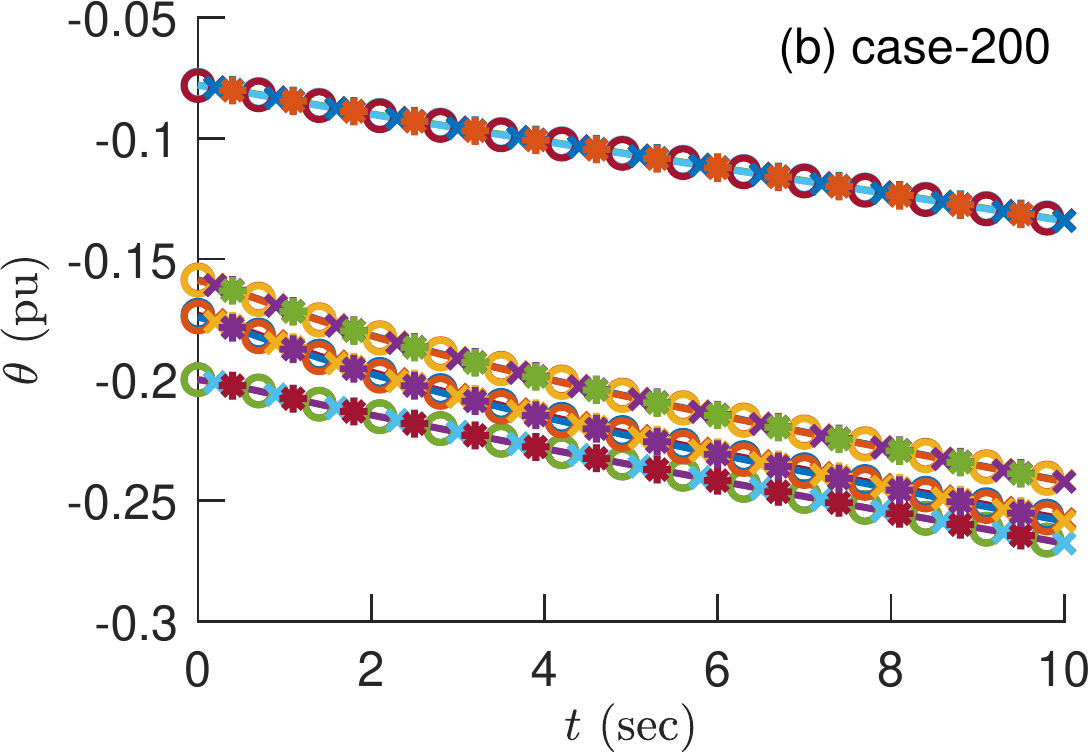}}{}{}
	\vspace{-0.4cm}
	\caption{Transient differential and algebraic state trajectories under load disturbance: (a) $\mr{case}$-$\mr{9}$ ($\alpha_{L} = 	2\%$) and (b) $\mr{case}$-$\mr{200}$ ($\alpha_{L} = 15\%$).}\label{fig:sim_cont}
\end{figure}

\begin{figure}[h]
	\centering
	\subfloat{\includegraphics[keepaspectratio=true,scale=1.2]{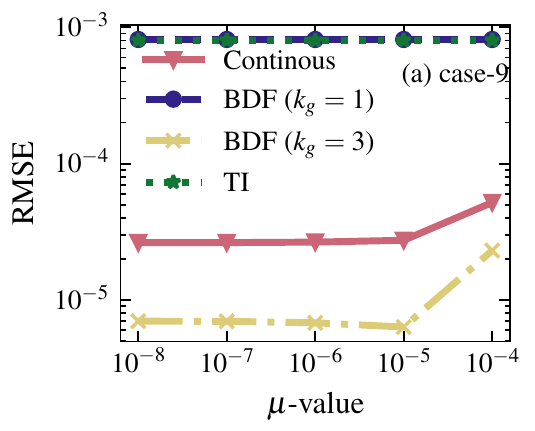}}\vspace{-0.3cm} \hspace{-0.3cm}
	\subfloat{\includegraphics[keepaspectratio=true,scale=1.2]{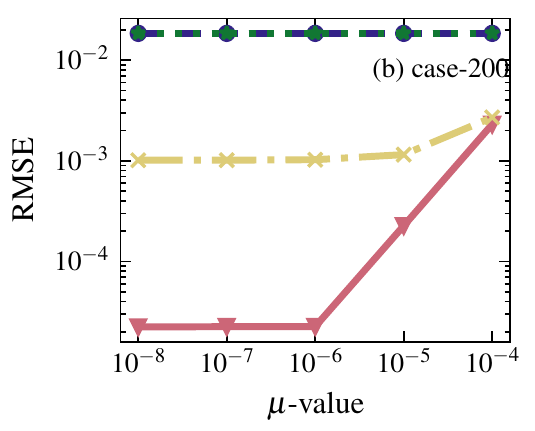}}
	\caption{RMSE on dynamic and algebraic state trajectories for the $\mr{Approx}$-$\mr{DAE}$ model while changing $\mu$-value.}\label{fig:mu}
\end{figure}

Regarding the transformed discrete-time models, the Jacobian of the differential equations remains the same, however, that of the algebraic constraints summarized in Table~\ref{tab:DAE-ODE_disc} differs. 
The Jacobian of the algebraic equations of the transformed systems are represented in Table~\ref{tab:Jac-DAE-ODE_disc}. For the $\mr{Approx}$-$\mr{DAE}$ system, the only difference is with the addition of a $\mu\eye_{n_a}$ perturbation term, where $\eye_{n_a}$ is an identity matrix of dimension similar to $\m{G}_{\m{x}_{a}}(\cdot)$. 

\begin{figure*}[h]
	\centering
	\hspace{-0.2cm}
	\subfloat{\includegraphics[keepaspectratio=true,scale=0.42]{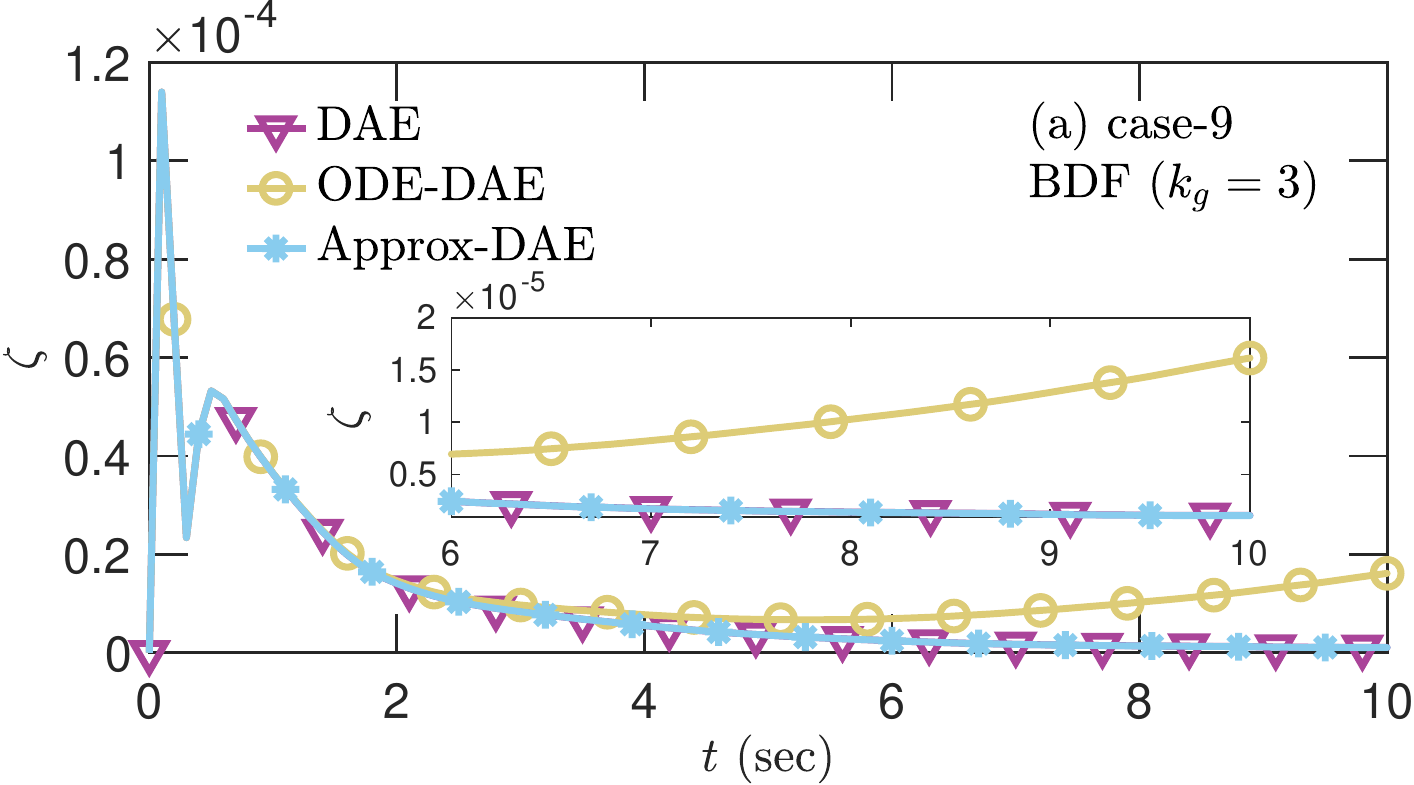}}\vspace{-.38cm}\hspace{0.2cm}
	\subfloat{\includegraphics[keepaspectratio=true,scale=0.42]{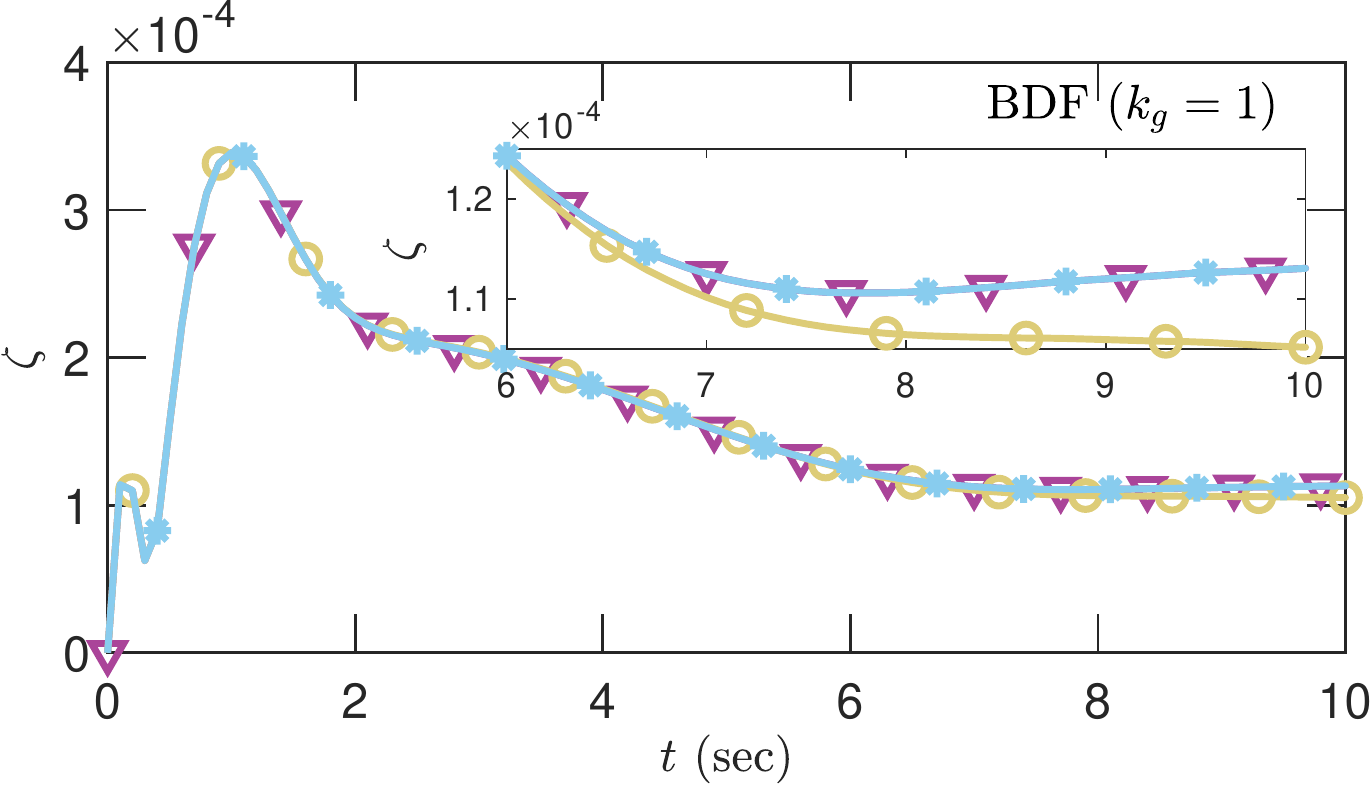}}\hspace{0.2cm}
	\subfloat{\includegraphics[keepaspectratio=true,scale=0.42]{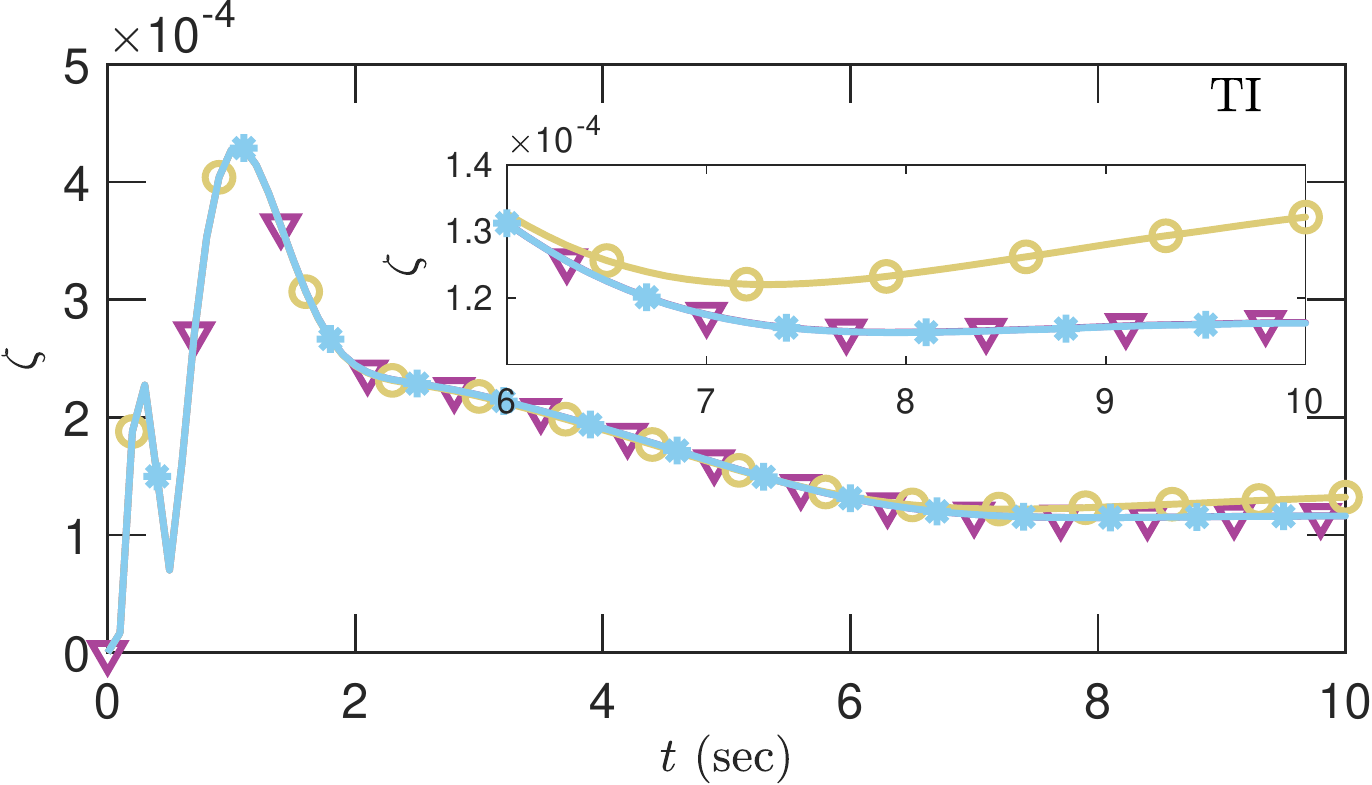}}{}{} \centering
	\subfloat{\includegraphics[keepaspectratio=true,scale=0.42]{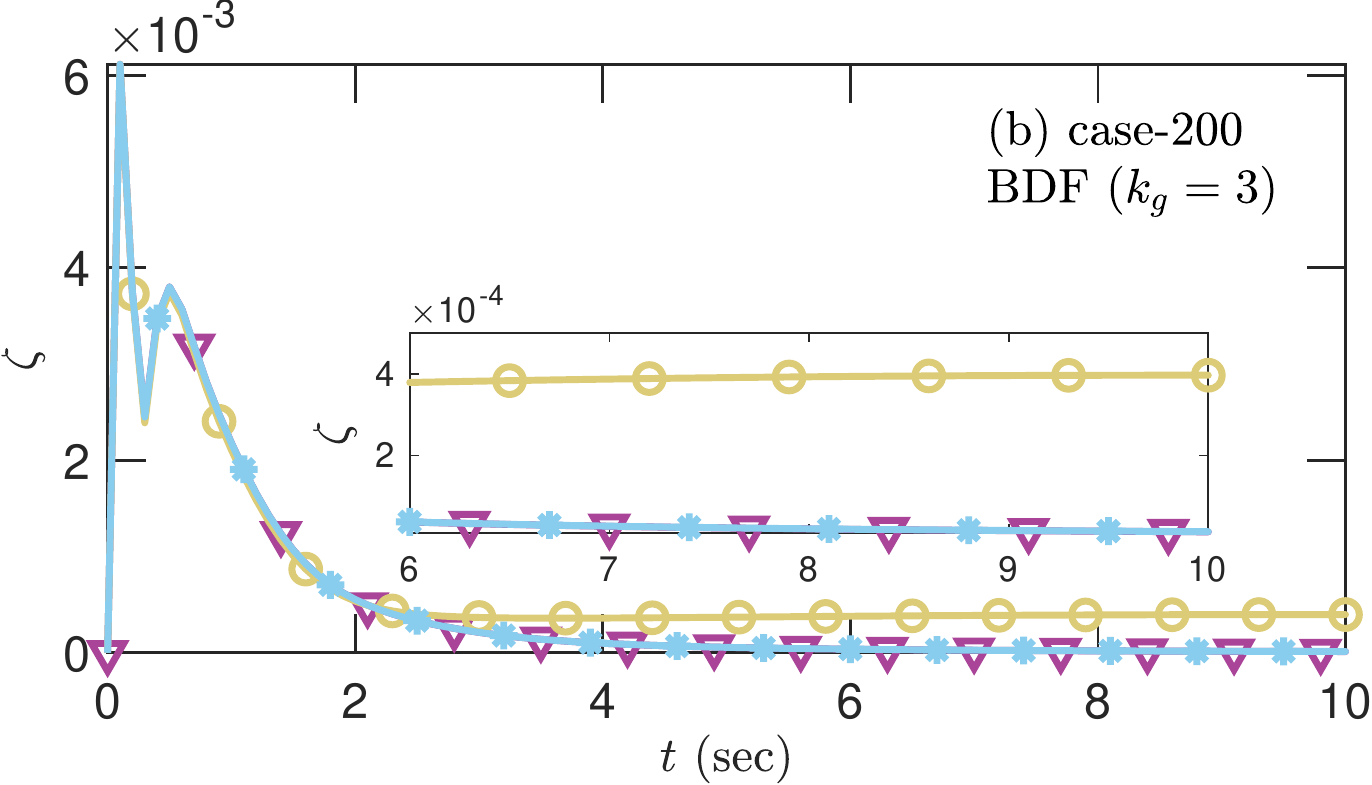}}
	\subfloat{\includegraphics[keepaspectratio=true,scale=0.42]{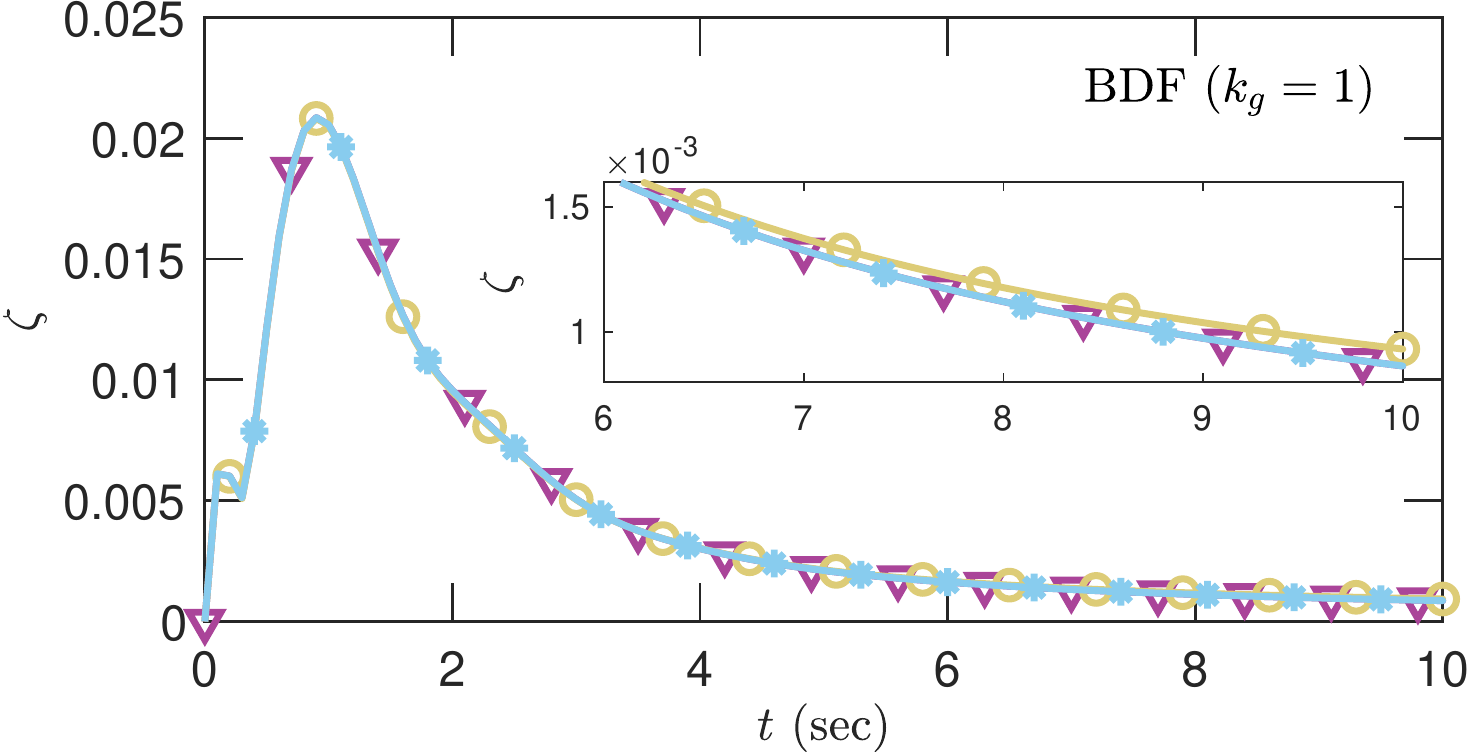}}
	\subfloat{\includegraphics[keepaspectratio=true,scale=0.42]{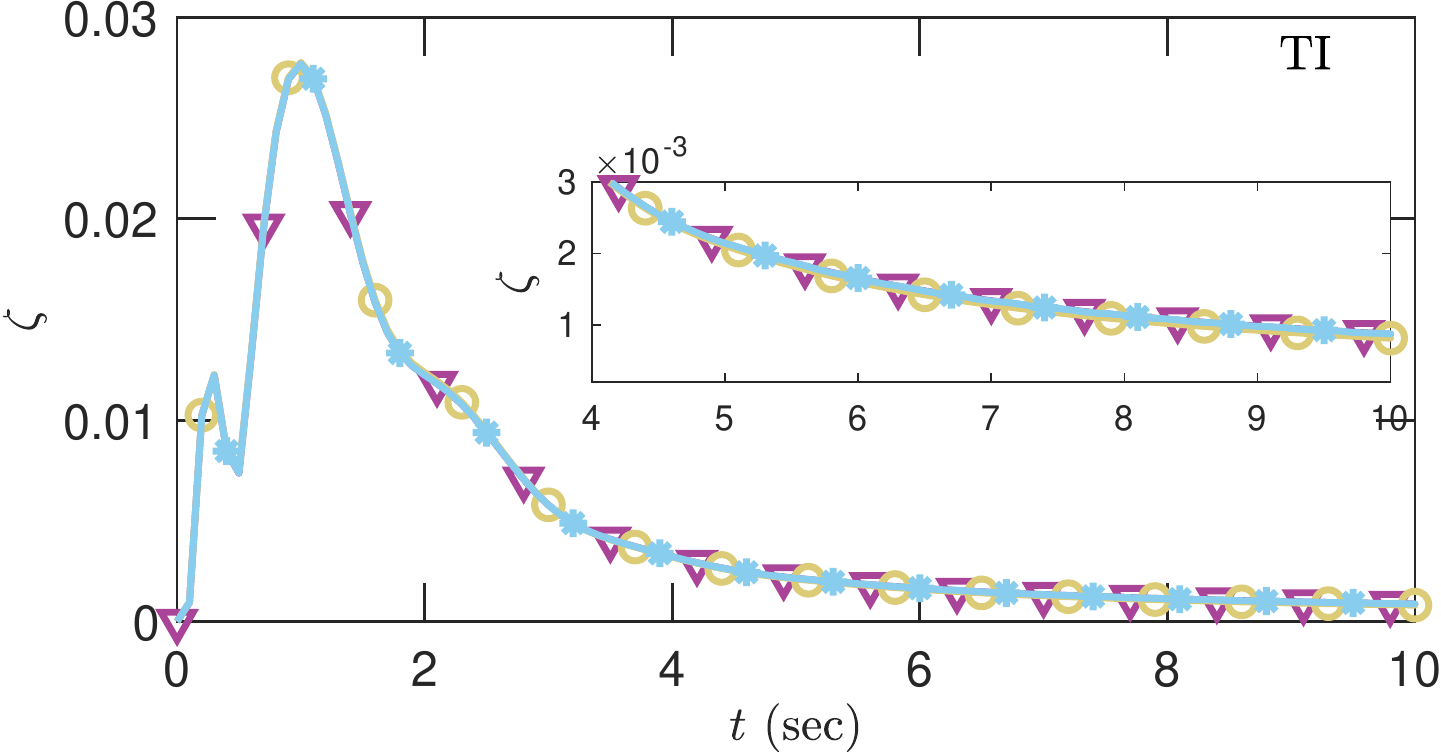}}
	\vspace{-0.3cm}
	\caption{Error norm on discrete-time domain state trajectories---under load disturbance---for the proposed model simulated under different discretization methods: (a) $\mr{case}$-$\mr{9}$ ($\alpha_{L} = 2\%$) and (b) $\mr{case}$-$\mr{200}$ ($\alpha_{L} = 15\%$).}\label{fig:sim-disc}
\end{figure*}

For the $\mr{ODE}$-$\mr{DAE}$ system, matrices $\widetilde{\m{G}}_{\m{x}_{d}}(\cdot) := \frac{\partial \tilde{\m{g}}(\m{x}_d,\m{x}_a,\m{u})}{\partial\m{x}_d} \in \mathbb{R}^{n_{a}\times n_{d}}$ and $\widetilde{\m{G}}_{\m{x}_{a}}(\cdot)  :=\frac{\partial\tilde{\m{g}}(\m{x}_d,\m{x}_a,\m{u})}{\partial\m{x}_a} \in \mathbb{R}^{n_{a}\times n_{a}}$ define the Jacobian of the algebraic equations~Table~\ref{tab:DAE-ODE_disc}, row 2 . For brevity, we define Jacobian $\tilde{\m{G}}_{\m{x}_{d}}(\cdot)$ in~\eqref{eq:G_tilde_d} and leave $\widetilde{\m{G}}_{\m{x}_{a}}(\cdot)$ for the reader to infer.


\begin{equation}~\label{eq:G_tilde_d}
	\begin{split}
		\hspace{-0.3cm} \widetilde{\m{G}}_{\m{x}_{d}}(\m{z}^{(i)}_{k}) := -(&-\m{G}_{\m{x}_{a}}^{-1} \tfrac{\partial \m{G}_{\m{x}_{a}}}{\partial \m{x}_{d}}\m{G}_{\m{x}_{a}}^{-1}  \m{G}_{\m{x}_{d}}\m{f}(\cdot) \\
		&+\m{G}_{\m{x}_{a}}^{-1} \tfrac{\partial \m{G}_{\m{x}_{d}}}{\partial \m{x}_{d}} \m{f}(\cdot)
		+ \m{G}_{\m{x}_{a}}^{-1} \m{G}_{\m{x}_{d}} \m{F}_{\m{x}_{d}}).
	\end{split}
\end{equation}

We note here that due to the existence of Jacobian matrices within the state-space formulation~\eqref{X_a_ode} of the IFT model, Hessian matrices $\m{G}_{\m{x}_{a,d}}:=\tfrac{\partial \m{G}_{\m{x}_{a}}}{\partial \m{x}_{d}}\in \mathbb{R}^{n_{a} \times n_{d}}$ and $\m{G}_{\m{x}_{d,d}} := \tfrac{\partial \m{G}_{\m{x}_{d}}}{\partial \m{x}_{d}}\in \mathbb{R}^{n_{a} \times n_{d}}$ appear under the NR iteration. 
To avoid the computationally convoluted evaluation of Hessian matrices, we refer to finite difference (FD) approximation techniques
	~\cite{Mickens1994,Venk2014}
	to approximate the matrices. FD approximations are based on Taylor series expansions of the system around a consistent operating point. For the approximations herein, we refer to the central difference approximation~\cite{Venk2014} which yields a lower error than the other difference methods. 
	As such $\m{G}_{\m{x}_{a,d}}$ can be approximated as~\eqref{eq:CDT} with difference variable $m \rightarrow 0$.
	\begin{equation}~\label{eq:CDT}
		\hspace{-0.3cm}	\m{G}_{\m{x}_{a,d}} = \hspace{-0.05cm} \tfrac{\m{G}_{\m{x}_{a}}(\m{x}^{(i)}_{d,k-m},\m{x}^{(i)}_{a,k}) -
			\m{G}_{\m{x}_{a}}(\m{x}^{(i)}_{d,k+m},\m{x}^{(i)}_{a,k})}{2m} + \m{O}(m^{2}),
	\end{equation}
	where $\m{G}_{\m{x}_{a,d}}$ represents the central difference approximation~\cite{Venk2014}, which yields a truncation error $\m{O}(m^{2})$.
\section{Case Studies}~\label{sec:case-studies}

The main objective of transient time-domain simulations is to trace the system's trajectory---after a disturbance---towards equilibrium~\cite{Lara2023}. Accordingly, we investigate the viability and effectiveness of performing both continuous- and discrete-time transient stability simulations of the proposed systems following a load disturbance.
Two networks of contrasting sizes are chosen: Western System Coordinating Council (WSCC) 9-Bus network ($\mr{case}$-$\mr{9}$), and ACTIVSg200-Bus network ($\mr{case}$-$\mr{200}$). 

The transient time-domain simulations are performed in MATLAB R2021b running on a Macbook Pro having an Apple M1 Pro chip with a 10-core CPU and 16 GB of RAM.
The baseline model herein refers to the nonlinear DAE~\eqref{eq:semi_NDAE_rep} and is simulated using MATLAB ODE/DAE solver $\mr{ode15s}$ and $\mr{ode15i}$. 
The settings chosen for the solvers are: $(i)$ absolute tolerance as $1\times 10^{-06}$, $(ii)$ relative tolerance as $1\times 10^{-05}$ and $(iii)$ maximum step size equal to $0.001$. For the transient discrete-time domain simulations, the discretization step size is set to $h = 0.1$ and the NR algorithm parameters are: $(i)$ absolute tolerance on $\mathcal{L}_2$--norm of iteration convergence as $10^{-2}$ and $(ii)$ maximum iterations as $10$.

Generator parameters are extracted from the power system toolbox (PST). Regulation and chest time constants for the generators are chosen as $R_{\mr{D}i} = 0.2 \; \mr{Hz/pu}$ and $T_{\mr{CH}i} = 0.2 \; \mr{sec}$. The steady state initial conditions for the power system are generated from the power flow solution obtained using MATPOWER. 
The synchronous speed is set to $\omega_{0} = 120\pi \;  \mr{rad/sec}$ and a power base of $100 \; \mr{MVA}$ is considered for the power system.

Starting from the initial steady state conditions, a load disturbance at $t > 0$ on initial load $(\mr{P}_\mr{L}^0,\mr{Q}_\mr{L}^0)$ is introduced. The perturbed magnitude under a load disturbance $(\alpha_{L})$ is computed as $(\mr{\tilde{P}}_\mr{L}^0,\mr{\tilde{Q}}_\mr{L}^0 ) = (1+\tfrac{\alpha_{L}}{100})(\mr{P}_\mr{L}^0,\mr{Q}_\mr{L}^0 )$. The load disturbance $(\alpha_{L})$ for $\mr{case}$-$\mr{9}$ is chosen within the range of $\{1 \%, 4\% \}$ with respect to the original loads; for test  $\mr{case}$-$\mr{200}$ the range is between $\{5 \%, 20\% \}$.

The continuous-time transient state trajectories representing algebraic state $\theta$ (bus angle) and differential state $\omega$ (generator synchronous rotor speed) for the DAE model and the proposed transformations are depicted in Fig.~\ref{fig:sim_cont}. The proposed models yield accurate transient state trajectory simulations as compared to the baseline DAE model for both test cases.

Before moving forward, we investigate the choice of $\mu$ on the accuracy of the $\mr{Approx}$-$\mr{DAE}$ system as compared with the baseline model. We calculate the root mean square error $(\mr{RMSE})$ of the transformed systems over time period $t$ as $\mr{RMSE} := \sqrt{\tfrac{\sum_{k=1}^{t} \m{e}_{k}^{2}}{t}}$, where $\m{e}_{k} := \abs{\m{\hat{x}}_{k} - \m{x}_{k}}$ is the difference between the states of the baseline DAE model~\eqref{eq:semi_NDAE_rep} denoted as $\m{x}_{k}$ and the states of proposed systems as $\m{\hat{x}}_{k}$. 

Fig.~\ref{fig:mu} presents RMSE for the $\mr{Approx}$-$\mr{DAE}$ model when varying the choice of $\mu$ between $\{1\times10^{-04}, 1\times10^{-08} \}$ under continuous and discrete time-domain simulations. It is illustrated that for both test cases, the choice of $\mu$ affects the error on state trajectories. For TI and BE discretization method, the choice of $\mu$ does not alter the RMSE error. Such methods are single step discretization methods which are inherently less accurate than the multi-step method, BDF, that is typically used for DAE systems. The RMSE for the continuous and BDF case is lower and is influenced by the choice of $\mu$ until it becomes asymptotic after a certain value of $\mu$. For the remainder of this work $\mu$ is chosen to be $1\times 10^{-06}$.


\begin{figure}[t]
	\centering
	\hspace{-0.1cm}
	\subfloat{\includegraphics[keepaspectratio=true,scale=0.72]{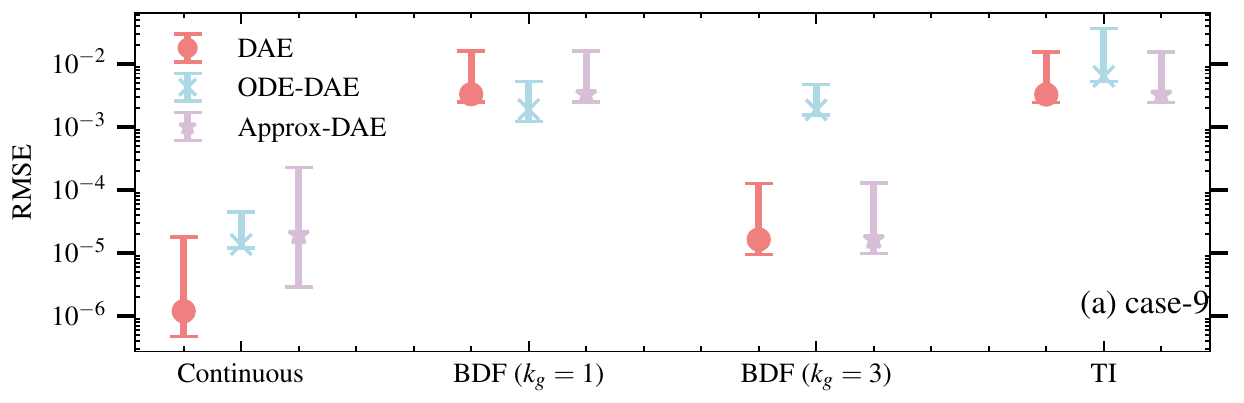}}\vspace{-0.4cm}{}{} 
	\subfloat{\includegraphics[keepaspectratio=true,scale=0.72]{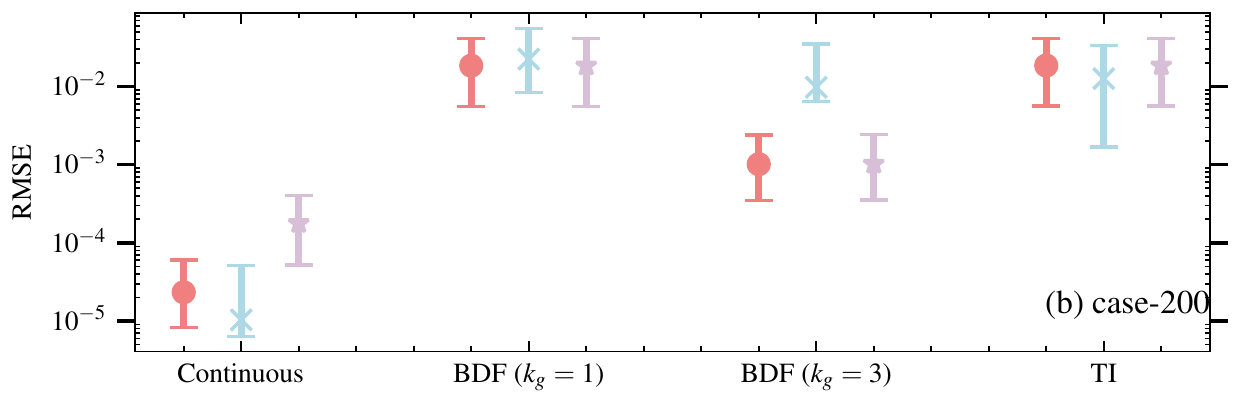}}{}{}
	\vspace{-0.45cm}
	\caption{RMSE on transient dynamic and algebraic states under varying load disturbances.}\label{fig:RMSE}
\end{figure}

To assess the accuracy of the discrete time-domain simulations, we compute the error norm $\zeta :=\norm{\hat{\m x}_k-{\m x}_k}_2$ between the baseline states $\m{x}_{k}$ at time and the states resulting from the proposed models $\m{\hat{x}}_{k}$. Fig.~\ref{fig:sim-disc} illustrates the error resulting from transient discrete time-domain simulation under the different discretization methods. The proposed models accurately depict the transient states; it can be noted that BDF method outperforms BE and TI for the DAE model and the proposed systems. To further investigate the applicability of the methods, we vary the load disturbance magnitude $\alpha_L$ under the aforementioned ranges. The RMSE for $\mr{case}$-$\mr{9}$  and $\mr{case}$-$\mr{200}$ under such load disturbances are depicted in Fig.~\ref{fig:RMSE}.

The performance of the transformed systems yields similar results to that of the DAE system under the continuous and discrete-time models while simulating the states under different transient conditions. We note that the performance under BDF discretization of the $\mr{ODE}$-$\mr{DAE}$ model, suggests that the BDF order $k_g$ might require changing to account for the altered stiffness in the dynamics. Refer to  Fig.~\ref{fig:time-sim} for the computational complexity resulting from simulating the proposed transformations relative to the nonlinear DAE system. The DAE and $\mr{Approx}$-$\mr{DAE}$ models require the same computational effort, however that of the $\mr{ODE}$-$\mr{DAE}$ model shows around a $100\%$ to $200\%$ increase in the computational complexity under discrete-time techniques. This is primarily due to the Hessian approximations~\eqref{eq:G_tilde_d}, whereby such an increase is not evident at the continuous time-domain level. An alternative Taylor series approximation for~\eqref{eq:G_tilde_d} can be chosen that requires less computational effort. Having provided the above results, the validity and accuracy of the proposed models are demonstrated, therefore proving their applicability for transient time-domain simulations of multi-machine power systems.

\section{Summary and Future Work}~\label{Sec:conclusion}
This paper presents transformations that result in solutions to nonlinear DAE power systems modeled as nonlinear ODEs. The validity and accuracy of the two methods have been investigated. The proposed transformations yield accurate depictions of the transient state-space models and can therefore be used in various feedback control or state estimation algorithms---our future work on this topic.
\begin{figure}[t]
	\centering
	\hspace{-0.1cm}
	\subfloat{\includegraphics[keepaspectratio=true,scale=0.72]{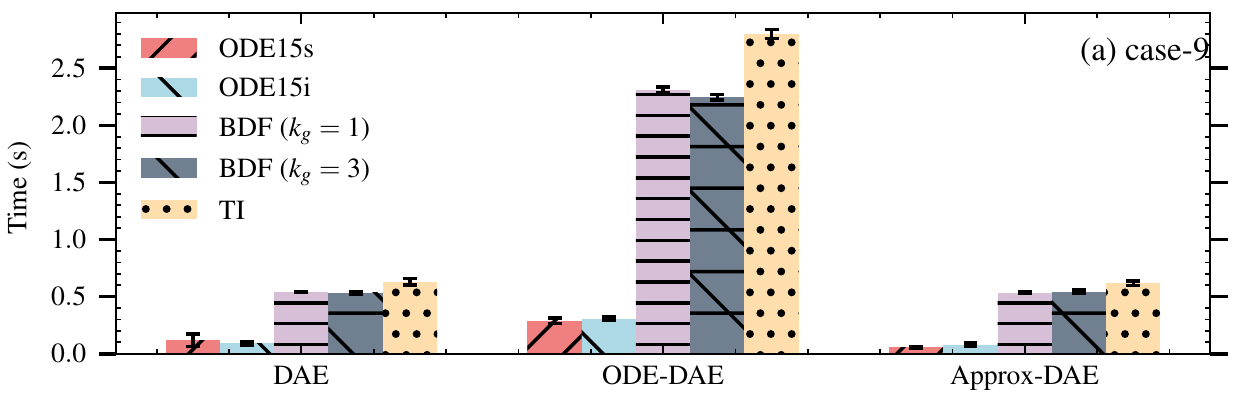}}\vspace{-0.45cm}{}{} 
	\subfloat{\includegraphics[keepaspectratio=true,scale=0.72]{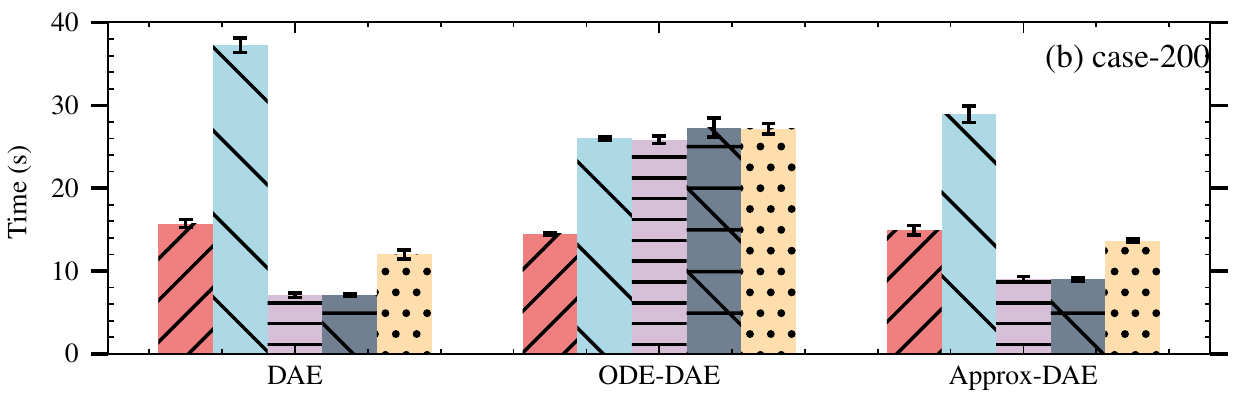}}{}{}
	\vspace{-0.45cm}
	\caption{Computational time for the different time-domain methods and power system models.}\label{fig:time-sim}
\end{figure}

\balance
\normalcolor
\bibliographystyle{IEEEtran}

\bibliography{IEEEabrv,DAE_TO_ODE.bib} 
\end{document}